\documentclass[12pt]{article}
\usepackage{eurosym}
\usepackage{natbib}
\usepackage{setspace}
\usepackage{amsmath,mathtools}
\usepackage{xcolor}
\usepackage{graphicx}
\usepackage{float}
\usepackage{subcaption}
\usepackage{hyperref}
\usepackage[top=2cm,bottom=4cm,left=3cm,right=3cm,asymmetric]{geometry}
\usepackage{amssymb}
\usepackage[labelfont=bf]{caption}

\newtheorem{proposition}{Proposition}

\newtheorem{lemma}{Lemma}
\newtheorem{corollary}{Corollary}

\let\emptyset\varnothing
\hypersetup{colorlinks = true, linkcolor=blue,urlcolor = blue,citecolor=blue}

\begin{document}

\setcounter{secnumdepth}{2}
\title{Audit the Auditors: Commitment versus Professional Judgment\thanks{The work described in this paper was partially supported by a grant from the
		Research Grants Council (RGC) of the Hong Kong Special Administrative Region,
		China [Project No. CityU 11500924]. We are grateful to the RGC panel and the three anonymous reviewers invited by the panel for their constructive comments and suggestions.}}
\author{Pingyang Gao\thanks{
HKU Business School, University of Hong Kong. Email: pgao@hku.hk}, Jinzhi Lu\thanks{Department of Accountancy, City University of Hong Kong. Email: jinzhilu@cityu.edu.hk}, and Zhenpeng Yang\thanks{Department of Accountancy, City University of Hong Kong. Email: zhenpyang2-c@my.cityu.edu.hk}}
\maketitle
\setlength{\baselineskip}{24pt} \setlength{\parindent}{15pt}	
\begin{abstract}
This paper provides a theoretical framework to evaluate the trade-off between the self-regulated peer review system and independent government inspection (PCAOB) in the auditing profession. We model the peer review system as a Judgment Regime, where a stakeholder utilizes professional expertise, captured as a private signal, to make ex-post decisions on verifying audit failures. In contrast, PCAOB inspection is modeled as a Commitment Regime, where the stakeholder lacks private information but can commit ex-ante to a predetermined level of verification. We find that the Judgment Regime benefits from a resource-allocation effect and a deterrence effect driven by informed verification, whereas the Commitment Regime deters audit failures through the first-mover advantage of ex-ante commitment. Our analysis demonstrates that the stakeholder prefers the peer review system if and only if the private signal is sufficiently informative. Furthermore, comparative statics reveal that higher verification costs or stronger audit incentives shift the stakeholder's preference toward PCAOB inspection.
\end{abstract}
\newpage
\section{Introduction}
\label{Section: Intro}

Auditors play a crucial role in ensuring the credibility and reliability of accounting information. As a result, the question of who audits the auditors is important. Before the Sarbanes-Oxley Act of 2002 (SOX), the auditing profession relied heavily on self-regulated peer review. SOX subsequently established the Public Company Accounting Oversight Board (PCAOB) and shifted the oversight of audits of public-company issuers toward independent public inspection. This regulatory shift remains a subject of ongoing debate, and prior empirical research has yielded mixed findings regarding the relative effectiveness of the two regimes \citep{lohlein2016peer}.

Despite the extensive empirical literature on peer review versus PCAOB inspection, there is a lack of theoretical work analyzing this regulatory shift. While the two oversight regimes differ along multiple dimensions in practice, our paper focuses on a fundamental yet underexplored trade-off: the tension between professional expertise and institutional commitment. Peer reviewers are typically practicing auditors who possess deep expertise in auditing standards and procedures, enabling them to form informed judgments about whether a particular engagement is deficient \citep[e.g.,][]{wallace1994exploratory, ehlen1996procedural}. PCAOB inspection, by contrast, operates through a highly institutionalized process with pre-specified procedures and predetermined resource allocations established before the inspected firm's conduct is observed \citep[e.g.,][]{defond2010should, lennox2010auditing, houston2013audit, lohlein2016peer}. This institutional structure enables PCAOB inspectors to credibly commit to a verification level before auditors make their effort choices. These institutional features suggest a trade-off: professional expertise may permit more informed ex-post judgment, whereas institutionalized procedures may support more credible ex-ante verification commitments.

To evaluate the trade-off between commitment and professional judgment, we start with a standard model in which an auditor exerts effort to attest to a firm's accounting report. The auditor may issue a clean opinion on a misstated accounting report, an event we define as an audit failure. The auditor's effort reduces the likelihood of such failure. However, the effort level is not directly observable. Thus, absent an effective oversight mechanism, the auditor has weak incentives to exert effort. We introduce a representative stakeholder who bears the loss associated with an undetected audit failure and chooses an oversight mechanism to identify such failures. The stakeholder can be interpreted broadly as representing the interests protected by audit oversight.

We compare two benchmark regimes that isolate the roles of professional judgment and ex-ante commitment. In the Judgment Regime (Regime J), which maps into the peer review system, the stakeholder's ability to identify the audit failure is aided by professional knowledge. In particular, the stakeholder receives a private signal regarding whether there is an audit failure. This signal assists the stakeholder in deciding how much resource to spend to formally identify the audit failure, if any. In the Commitment Regime (Regime C), which maps to PCAOB inspection, there is no private information available to the stakeholder. Instead, the stakeholder can commit a predetermined amount of resources towards identifying the audit failure. Although peer review and PCAOB inspection each involve both professional judgment and institutional commitment in practice, we model them as polar benchmark regimes to isolate the fundamental trade-off between informed ex-post judgment and credible ex-ante commitment.

We find that the stakeholder's private information in Regime J plays two roles. First, it improves resource allocation: the stakeholder's private signal allows him to spend more (less) resources in identifying the audit failure when he knows that the failure is more (less) likely (resource allocation effect). Second, it also plays a deterrence role: the auditor exerts more effort when the private signal becomes more informative. In contrast, the stakeholder's commitment in Regime C only has a deterrence effect: the auditor exerts more effort due to the stakeholder's first-mover advantage, a mechanism similar to that in the sequential oligopoly market \citep{von1934marktform}.

The stakeholder's preference between the two regimes is determined by the trade-off among the three forces described above. We find that when the private signal in Regime J is uninformative, the stakeholder prefers Regime C due to the deterrence effect of commitment (i.e., first-mover advantage). When the private signal in Regime J is perfectly informative, the stakeholder prefers Regime J, because the resource-allocation and deterrence benefits of private information outweigh the deterrence effect of commitment. In addition, we find that the stakeholder's expected payoff in Regime J monotonically increases in the informativeness of the private signal. Thus, the stakeholder prefers Regime J over Regime C if and only if the private signal is sufficiently informative.  

Our comparative statics analysis reveals how changes in verification and audit incentives shift the trade-off between the two regimes. First, as the cost of verification increases (or its benefit decreases), the Judgment Regime becomes less appealing to the stakeholder. Higher costs undermine the value of ex-post private information because the stakeholder becomes reluctant to verify suspected audit failures. This reluctance weakens both the resource allocation and deterrence effects of the private signal. While the commitment in Regime C also loses value as deterrence weakens, this loss is smaller in magnitude than the overall loss suffered by the Judgment Regime.

Second, an increase in the audit incentive, meaning it is easier to motivate the auditor, makes the Commitment Regime relatively more attractive. Although both regimes benefit from higher audit incentives, the deterrence power of private information in the Judgment Regime is partially offset. This occurs because the stakeholder anticipates the increased audit effort and consequently reduces his own verification level. In contrast, the Commitment Regime allows the stakeholder to fully capitalize on the auditor's increased responsiveness, making ex-ante commitment a more powerful and valuable tool.

\subsection{Related Literature}
\label{SubSection: Contributions}
Our paper contributes to the literature on peer review and PCAOB inspection. Existing empirical research suggests that both peer review and PCAOB inspection can enhance audit quality, although the evidence regarding their relative effectiveness and how market participants recognize their effects remains mixed \citep{lohlein2016peer}. Recent studies further document that PCAOB inspections can improve reporting credibility and accountability through mechanisms such as stronger enforcement, greater public oversight, and partner identification requirements \citep{defond2010should, burke2019audit, gipper2020public}. Our paper complements this empirical literature by providing a theoretical framework for comparing peer review and PCAOB inspection through the trade-off between professional judgment and ex-ante commitment. We discuss the implications in detail in Section \ref{Section: Statics}.

Our paper contributes to the analytical literature on audit regulation \citep[e.g.,][]{dye1993auditing, ye2013economics, chen2019effects}. Specifically, we answer the call by \citet{ye2023theory} for further research on the role of active enforcers. In our model, PCAOB inspectors do more than merely exert effort to uncover audit failures; they also incentivize auditors through ex-ante commitments. Furthermore, our work complements existing research on the broader consequences of the Sarbanes-Oxley Act (SOX) \citep[e.g.,][]{patterson2007effects, deng2012auditors} and the effects of regulatory tightening. For instance, while \citet{gao2019auditing} find that stricter auditing standards can restrict auditors' professional judgment and discourage expertise acquisition, our study takes a different approach by focusing on the professional judgment of peer reviewers. Our paper stands apart by examining a distinct aspect of SOX: the fundamental regulatory shift from a peer-review system to independent PCAOB inspection.

The rest of the paper proceeds as follows. Section \ref{Section: Model} describes the model. Section \ref{Section: Equilibrium} solves the equilibrium decisions in the two regimes. Section \ref{Section: MainAnalysis} shows our main results. Section \ref{Section: AuditorandRisk} evaluates the two regimes in terms of the audit risk. Section \ref{Section: Statics} presents the comparative statics analysis and empirical implications. Section \ref{Section: Conclusions} concludes.

\section{Model Set-Up}
\label{Section: Model}
We augment a standard audit model with the stakeholder's verification of audit reports. A firm has access to a project that requires an initial investment of $I$. The project's underlying state is $\theta$, which is either good (G) or bad (B) with equal probability. The project delivers a cash flow of $X>0$ in the good state and $0$ in the bad state. We assume $\Pr(\theta=G)X>I$ such that the firm's default action is to invest in the absence of additional information. The firm does not have private information about $\theta$ and always sends the auditor a favorable attestation report. The model consists of two main players, an auditor and a representative stakeholder. The auditor first conducts an audit of the firm's initial favorable report. Then, the stakeholder decides the level of effort to verify the correctness of the audit opinions. Finally, the firm makes the investment decision based on the audit and verification results. We now discuss the model ingredients in more detail.

The auditor chooses an audit effort $\beta$ at a cost $K(\beta)=\frac{1}{2}k\beta^2$ to obtain a signal about the true state $\theta$, and then issues it truthfully as an audit report $r\in\{g,b\}$. $r=g$ represents an unqualified (clean) opinion, whereas $r=b$ is a qualified opinion that rejects the firm's favorable assessment. The effectiveness of this audit technology depends on the auditor's effort $\beta$. Specifically, the report $r$ satisfies
\begin{align}
\label{AuditTech}
     \Pr(r=g \mid \theta=G,\beta)&=1,\\
    \Pr(r=g \mid \theta=B,\beta)&=1-\beta.
\end{align} 
We define \textit{audit failure} as the event where the auditor issues a clean audit opinion in the bad state. 

After receiving the auditor's report, the stakeholder can verify the true state. Upon receiving a negative report from the auditor $(r=b)$, there is no need for the stakeholder to conduct further verification, as the qualified opinion itself confirms the bad state and the firm does not invest. Upon receiving a clean opinion ($r=g$), the stakeholder can identify the state $\theta$ with a chosen probability $a\in[0,1)$. Conducting verification incurs a cost of $C(a) = \frac{1}{2}ca^2$ for the stakeholder, where $c > 0$. If verification under a clean opinion reveals a bad state ($\theta=B$), the auditor will suffer a loss $L>0$ and the firm will not invest in the project. If verification confirms a good state ($\theta=G$) or fails to detect a bad state ($\theta=B$), the firm will invest in the project. In these cases, the auditor's payoff is normalized to 0. We assume the stakeholder represents the interests of investors, and his payoff is $X-I$ if the project is carried out in the good state, $-I$ if the project is carried out in the bad state, and $0$ if the project is not carried out. We assume $c>I>0$ and $k>L>0$ to ensure interior levels of audit effort and verification. We use $l\equiv L/k$ to denote the auditor's audit incentive and $i\equiv I/c$ to denote the stakeholder's verification incentive.

We interpret the stakeholder as a representative investor-protection principal. The stakeholder's objective is to prevent capital from being allocated to value-destroying projects while allowing value-creating projects to proceed. The firm is modeled as a passive decision maker that mechanically invests when the available audit and verification information supports investment. Thus, the stakeholder's payoff equals the net payoff to investors from the resulting investment decision: $X-I$ when a good project is undertaken, 
$-I$ when a bad project is undertaken, and zero when investment is prevented.

The stakeholder can choose between two mechanisms, modeled as two regimes, to identify audit failure. In practice, peer reviewers may also possess some commitment power, and PCAOB inspectors may also have access to private information. We focus on the two polar cases, i.e., pure judgment without commitment and pure commitment without private information, to cleanly isolate the fundamental trade-off between these two forces.

In the \textbf{Judgment Regime (Regime J)}, the stakeholder receives an additional private signal $s$ about the state. The stakeholder uses both the audit report $r$ and the private signal $s$ as his information set to form posterior beliefs. In our main analysis, we let $s$ be a general signal with finite support. To study how the informativeness of the private signal affects the equilibrium, we let the stakeholder's private signal belong to a continuous nested family of information structures $\{\pi_q\}_{q\in[0,1]}$. The parameter $q$ indexes informativeness such that, whenever $q'>q$, $\pi_{q'}$ strictly Blackwell dominates $\pi_q$. Moreover, $\pi_0$ is uninformative, $\pi_1$ is perfectly informative, and the conditional signal distributions vary continuously in $q$.

The timeline for Regime J is summarized as follows:
\begin{itemize}
    \item At date 1, the stakeholder takes no action.
    \item At date 2, the auditor exerts effort $\beta$ to issue an audit report $r$.
    \item At date 3, the stakeholder receives $s$ and chooses the verification $a$ based on $s$ and the auditor's report $r$. The true state is identified with probability $a$. The payoffs are realized.
\end{itemize}

In the \textbf{Commitment Regime (Regime C)}, the stakeholder can make an ex-ante commitment to the level of resources used to verify the deficiency before the auditor exerts effort. We assume the stakeholder only conducts verification conditional on a clean audit report.

The timeline for Regime C is as follows:
\begin{itemize}
    \item At date 1, the stakeholder commits to $a$.
    \item At date 2, the auditor observes the committed verification level and exerts effort $\beta$.
    \item At date 3, the stakeholder conducts the committed verification $a$ when $r=g$. The payoffs are realized.
\end{itemize}

Our equilibrium concept is Perfect Bayesian Equilibrium (PBE). A PBE consists of the auditor's audit effort $\beta_R$ and the stakeholder's verification strategy $a_R(\Omega_R)$, where $R\in\{J,C\}$ denotes the Regime J or C, and $\Omega_R$ denotes the stakeholder's information set.\footnote{Specifically, the information set for Regime J consists of the private signal and the audit report, i.e., $\Omega_J=\{s,r\}$. In contrast, the information set in Regime C includes only the audit report, i.e., $\Omega_C=\{r\}$.} A PBE requires that (1) both the auditor and the stakeholder maximize their respective objective functions, given their beliefs and the strategies of others; and (2) the stakeholder uses Bayes' rule to update beliefs about the state $\theta$.

\section{Equilibrium}
\label{Section: Equilibrium}
\subsection{Judgment Regime}
\label{SubSection: EquiJ}
We first solve the stakeholder's response given the conjecture about the auditor's effort. We then determine the auditor's optimal audit effort and finally impose rational expectations to solve the equilibrium.

Suppose that the stakeholder conjectures the auditor's effort to be $\hat{\beta_J}$. The stakeholder's posterior belief conditional on any realization of $s$ (denoted by $\omega_t$) can be computed by Bayes' rule:
\begin{align*}
&\Pr(\theta=B|r=g,s=\omega_t)=\frac{\Pr(r=g,s=\omega_t|\theta=B)\Pr(\theta=B)}{\Pr(r=g,s=\omega_t)}\\
&=\frac{(1-\hat{\beta_J})\Pr(s=\omega_t|\theta=B)}{(1-\hat{\beta_J})\Pr(s=\omega_t|\theta=B)+\Pr(s=\omega_t|\theta=G)}.
\end{align*}
Thus, conditional on $r=g$ and $s=\omega_t$, the stakeholder's optimization problem at date 3 is 
\begin{align*}
\max_a\;\Pr(\theta=G|r=g,s=\omega_t)(X-I)+(1-a)\Pr(\theta=B|r=g,s=\omega_t)(-I)-\frac{1}{2}ca^2.
\end{align*}
The first-order condition suggests that
\begin{align}
a(\omega_t)=\frac{I}{c}\Pr(\theta=B|r=g,s=\omega_t). \label{aw}
\end{align}

Equation \eqref{aw} is the stakeholder's optimal verification level when receiving $(r=g,s=\omega_t)$. Considering all the realizations of $s$, the expected level of verification conditional on $\theta=B$ and $r=g$ is
\begin{align}
\bar{a}=\sum_t \Pr(s=\omega_t|\theta=B,r=g)a(\omega_t), \label{bara}
\end{align}
where $\Pr(s=\omega_t|\theta=B,r=g)=\frac{\Pr(s=\omega_t,r=g|\theta=B)}{\Pr(r=g|\theta=B)}=\Pr(s=\omega_t|\theta=B)$, the last equality following from conditional independence ($s\perp r\mid\theta$). For ease of exposition, $\bar{a}$ is defined as the \textit{critical verification level}.

The auditor's optimization problem is thus
\begin{align*}
	\max_{\beta}\; -\Pr(\theta=B)(1-\beta)\bar{a}L-\frac{1}{2}k\beta^2.
\end{align*}
The first-order condition suggests that
\begin{align}
	\beta=\frac{\Pr(\theta=B)L}{k}\bar{a}. \label{betaJ}
\end{align}

Plugging \eqref{bara} into \eqref{betaJ} and imposing the rational expectations requirement $\beta=\hat{\beta_J}$ leads to 
\begin{align}
	\beta-\frac{\Pr(\theta=B)L}{k}\sum_t \Pr(s=\omega_t|\theta=B)a(\omega_t)=0 \label{equilibriumJ},
\end{align}
which is an equation of $\beta$ that determines the equilibrium. Denoting the LHS of \eqref{equilibriumJ} as $F_J(\beta)$, we have:
\begin{align}
	F_J(\beta)=\beta-\frac{\Pr(\theta=B)LI}{kc}\sum_t \frac{(1-\beta)[\Pr(s=\omega_t|\theta=B)]^2}{(1-\beta)\Pr(s=\omega_t|\theta=B)+\Pr(s=\omega_t|\theta=G)}. \label{FJbeta}
\end{align}
We can easily observe that $F_J(0)<0$, $\lim_{\beta\to1^-}F_J(\beta)>0$, and $F'_J(\beta)>0$. Therefore, by the Intermediate Value Theorem and the monotonicity of $F_J(\beta)$, there must be a unique solution for the equation \eqref{equilibriumJ}, indicating the existence and the uniqueness of the equilibrium in Regime J. 

\begin{proposition} \label{Prop: equilibriumJ}
There exists a unique equilibrium in Regime J.  The stakeholder's verification level and the auditor's effort level are uniquely determined by equations \eqref{aw} and \eqref{equilibriumJ}. 
\end{proposition}

We now analyze the comparative statics of the equilibrium outcomes under Regime J.  A key question is how the informativeness of the private signal, $q$, affects the equilibrium outcomes. We use $\beta_J(q)$ to denote the equilibrium audit effort level as a function of informativeness $q$, and $\bar{a}(\beta_J,q)$ to denote the critical verification level as a function of $\beta_J$ and $q$.
\begin{proposition} \label{Prop: beta}
\begin{enumerate}
	\item Keeping $\beta_J$ constant, the critical verification level increases in the informativeness of $s$: $\frac {\partial \bar{a}(\beta_J,q)}{\partial q}>0$.
	\item The auditor's effort increases in $q$: $\frac {d\beta_J(q)}{d q}>0$. 
    \item The critical verification level increases in $q$: $\frac {d\bar{a}(\beta_J(q),q)}{d q}>0$.
\end{enumerate}
\end{proposition}

Proposition \ref{Prop: beta} highlights the role of the stakeholder's private information. First, it allows the stakeholder to choose the verification level based on the realization of $s$. Part 1 shows that a more informative signal helps the stakeholder make more informed verification decisions. In particular, the stakeholder with a more informative signal can intensify verification efforts when the underlying state is bad, while conserving resources when the state is good. This resource allocation improvement results in a higher critical verification level given any fixed level of audit effort. 

Second, the stakeholder's private information has a deterrence power because it motivates higher auditor effort (Part 2). In Regime J, the informativeness of the private signal is common knowledge. The auditor expects that the stakeholder with a more informative signal will conduct more intensive verification when the true state is bad. If the auditor reduces audit effort, a clean audit report becomes more likely, yet a deficiency is also more likely to be detected. This, in turn, increases the probability that the auditor will be penalized. In anticipation of this outcome, the auditor is motivated to exert higher effort to avoid punishment.

Part 3 captures the overall effect of a more informative signal on verification. To formalize this effect, we examine the derivative of the critical verification level with respect to the signal's informativeness. 
\begin{align}\label{Lemma:beta decomp}
\frac{d\bar{a}(\beta_J(q),q)}{d q} = \underbrace{\frac{\partial\bar{a}(\beta_J(q),q)}{\partial q}}_{+\ \text{direct effect}} + \underbrace{\underbrace{\frac{\partial\bar{a}(\beta_J(q),q)}{\partial \beta_J(q)}}_{-} \times \underbrace{\frac{d\beta_J(q)}{dq}}_{+}}_{-\ \text{indirect effect}}
\end{align}
The direct effect is exactly Part 1. The indirect effect consists of two components. The first is negative since the stakeholder would verify less intensively when anticipating a higher audit effort. The second component is positive, following Part 2. While the direct one is positive and the indirect is negative, Part 3 of Proposition \ref{Prop: beta}  demonstrates the dominance of the direct effect in equilibrium.

The next corollary explores the effect of audit incentive and verification incentive on audit effort and critical verification level. 
\begin{corollary} \label{coroJ}
Use $l\equiv L/k$ to denote the auditor's audit incentive and $i\equiv I/c$ to denote the stakeholder's verification incentive. We have $\frac{\partial \beta_J}{\partial l}>0$, $\frac{\partial \bar{a}}{\partial i}>0$, $\frac{\partial \beta_J}{\partial i}>0$, and $\frac{\partial \bar{a}}{\partial l}<0$.
\end{corollary}

$\frac{\partial \beta_J}{\partial i}>0$ indicates the role of the stakeholder's verification in motivating the auditor to exert more audit effort. An increase in the verification incentive increases the level of critical verification faced by the auditor, which in turn increases the probability of being punished for issuing an incorrect clean report. Given this pressure, the auditor is induced to exert more audit effort to avoid an audit failure. 

The intuition for $\frac{\partial \bar{a}}{\partial l}<0$ is as follows. As the audit incentive $l$ increases, the stakeholder expects that the auditor will exert more audit effort. Such an expectation decreases the stakeholder's posterior belief for the bad state, thus causing the critical verification $\bar{a}$ to decrease. 

\subsection{Commitment Regime}
\label{SubSection: EquiC}
At date 2, the auditor solves the following optimization problem after observing the verification level $a$ committed by the stakeholder
\begin{align*}
    \max_{\beta} \quad -\Pr(\theta=B)(1-\beta)aL-\frac{1}{2}k\beta^2.
\end{align*}
The first-order condition implies 
\begin{equation}
    \label{expressionbetaC}
    \beta(a)=\frac{\Pr(\theta=B)L}{k}a.
\end{equation}

Equation \eqref{expressionbetaC} illustrates that a higher level of verification motivates the auditor to exert more effort. 

At date 1, the stakeholder chooses the verification level with the anticipation of \eqref{expressionbetaC}. Thus, the optimization problem is
\begin{align*}
    \max_{a} \quad &\Pr(\theta=G) (X-I)-\Pr(\theta=B)(1-\beta(a)) (1-a) I\\&-\left(\Pr(\theta=B)(1-\beta(a))+\Pr(\theta=G)\right)\frac{ca^2}{2}.
\end{align*}
The first-order condition w.r.t. $a$ suggests that 
\begin{align}
	3l a^{2} - (8 + 4il)a + 2i(l + 2) = 0, \label{equilibriumC}
\end{align}
where $i\equiv \frac{I}{c}$ and $l\equiv \frac{L}{k}$ represent the verification incentive and audit incentive, respectively. In the proof of Proposition \ref{Prop: equilibriumC}, we show that \eqref{equilibriumC} has a unique solution. We use $\beta_C$ and $a_C$ to denote the auditor's effort level and the stakeholder's verification level in equilibrium. We summarize the discussion in the following proposition. 
\begin{proposition} \label{Prop: equilibriumC}
There exists a unique equilibrium in Regime C. $\beta_C$ and $a_C$ are determined by equations \eqref{expressionbetaC} and \eqref{equilibriumC}. 
\end{proposition}

The next corollary explores the effect of audit incentive and verification incentive on audit effort and verification level. 
\begin{corollary} \label{coroC}
Use $l\equiv L/k$ to denote the auditor's audit incentive and $i\equiv I/c$ to denote the stakeholder's verification incentive.  We have $\frac{\partial \beta_C}{\partial i}>0$, $\frac{\partial a_C}{\partial i}>0$, $\frac{\partial \beta_C}{\partial l}>0$, and $\frac{\partial a_C}{\partial l}>0$. 
\end{corollary}

Similar to Regime J, an increase in the verification incentive ($i$) raises the level of verification faced by the auditor, thereby increasing the probability of being punished for issuing an incorrect clean report. Given this pressure, the auditor is induced to exert more audit effort to avoid an audit failure. This explains $\frac{\partial a_C}{\partial i}>0$ and $\frac{\partial \beta_C}{\partial i}>0$.

More interestingly, an increase in the auditor's incentive ($l$) also drives up both verification and effort. Because the stakeholder moves first, he recognizes that a higher $l$ makes the auditor more sensitive to the threat of inspection. Unlike in Regime J where the stakeholder reduces the verification level, the stakeholder in Regime C leverages this sensitivity. He commits to a higher level of verification to maximize deterrence, extracting an even greater increase in auditor effort. This explains $\frac{\partial \beta_C}{\partial l}>0$ and $\frac{\partial a_C}{\partial l}>0$.

\section{Main Results}
\label{Section: MainAnalysis}
Having established the equilibrium, we compare the stakeholder's expected utility across the two regimes. As shown in the proof of Lemma \ref{Lemma: ExtremeSignals}, the stakeholder's expected payoff in Regime J can be written out as a function of auditor effort and the critical verification level, thus denoted as $U_{J}(\beta_J(q),\bar{a}(\beta_J(q),q))$. We use $U_{C}(\beta_C,a_C)$ to denote the stakeholder's expected payoff in Regime C. The next proposition presents our main result.

\begin{proposition} \label{Prop:trade-offStakeinq}
	There exists a threshold $q^{\dagger}\in(0,1)$ such that the stakeholder prefers Regime J over Regime C if and only if $q>q^{\dagger}$.
\end{proposition}

To build the intuition for Proposition \ref{Prop:trade-offStakeinq}, the following lemma establishes the stakeholder's preference at the extremes of signal informativeness, as well as the monotonicity of his expected utility in Regime J with respect to $q$.

\begin{lemma}\label{Lemma: ExtremeSignals}
	\begin{enumerate}
		\item When $s$ is uninformative, the stakeholder prefers Regime C over Regime J. That is, $U_C(\beta_C,a_C)>U_{J}(\beta_J(q=0),\bar{a}(\beta_J(q=0),q=0))$.
		\item The stakeholder prefers Regime J with a more informative $s$. That is, $\frac{dU_J(\beta_J(q),\bar{a}(\beta_J(q),q))}{dq}>0$.
		\item When $s$ is perfectly informative, the stakeholder prefers Regime J over Regime C. That is, $U_C(\beta_C,a_C)<U_{J}(\beta_J(q=1),\bar{a}(\beta_J(q=1),q=1))$.
	\end{enumerate}
\end{lemma}

Part 1 of Lemma \ref{Lemma: ExtremeSignals} highlights the important role of ex-ante commitment (hereafter referred to as the \textit{deterrence effect of commitment}). When verification occurs ex post, the stakeholder lacks an incentive to verify a clean report, based on the rational belief that a clean report is more likely to arise from a good state.  Anticipating this, the auditor exerts low audit effort. A credible ex-ante commitment incentivizes auditor effort and shifts the equilibrium from low verification and low effort to one characterized by high verification and high audit effort. We confirm this intuition in the proof of Lemma \ref{Lemma: ExtremeSignals} by showing that $a_C>\bar{a}(\beta_J(q=0),q=0)$ and $\beta_C>\beta_J(q=0)$. Furthermore, Part 1 demonstrates that the stakeholder is ultimately better off in this high-verification, high-effort equilibrium, even after accounting for the greater resources devoted to verification. Overall, the deterrence effect of commitment is captured by 
\begin{align} \label{Eqn:commit}
\Delta_1\coloneqq U_{J}(\beta_J(q=0),\bar{a}(\beta_J(q=0),q=0))-U_{C}(\beta_C,a_C)<0.
\end{align}

Part 2 of Lemma \ref{Lemma: ExtremeSignals} shows that the stakeholder prefers a more informative signal in Regime J. This is due to two effects. The first is the \textit{resource allocation effect}. Keeping the auditor's effort fixed, a more informative signal enables a better allocation of resources based on the signal (i.e., a more informed choice of $a$). This improved resource allocation, reflected as an increase in the critical verification level, improves the stakeholder's ex-ante expected utility, holding the level of audit effort fixed. The resource allocation effect is captured by 
\begin{align} \label{Eqn:resource}
\Delta_2\coloneqq U_{J}(\beta_J(q=0),\bar{a}(\beta_J(q=0),q))-U_{J}(\beta_J(q=0),\bar{a}(\beta_J(q=0),q=0))>0.
\end{align}

The second is the deterrence effect of private information. As shown in Parts 2 and 3 of Proposition \ref{Prop: beta}, a more informative signal induces greater audit effort, and the increase in audit effort affects the stakeholder's utility in two opposing ways. On the one hand, it reduces the likelihood of audit failure. On the other hand, it lowers the stakeholder's verification levels for all realizations of $s$. Intuitively, the lower levels of verification reflect the stakeholder's optimal response to higher auditor effort. Although the reduction in verification decreases utility because audit failures are less likely to be uncovered (indirect), this loss is outweighed by the gain from the reduced incidence of audit failures (direct). The deterrence effect of private information is captured by 
\begin{align} \label{Eqn:deter}
\Delta_3\coloneqq U_{J}(\beta_J(q),\bar{a}(\beta_J(q),q))-U_{J}(\beta_J(q=0),\bar{a}(\beta_J(q=0),q))>0.
\end{align}

Notice that the difference between $U_{J}$ and $U_{C}$ can be decomposed into the three terms above:
\begin{align*}
	&\,U_{J}(\beta_J(q),\bar{a}(\beta_J(q),q))-U_{C}(\beta_C,a_C):=\Delta_1+\Delta_2+\Delta_3\\
	=&\,U_{J}(\beta_J(q=0),\bar{a}(\beta_J(q=0),q=0))-U_{C}(\beta_C,a_C)\\
	+&\,U_{J}(\beta_J(q=0),\bar{a}(\beta_J(q=0),q))-U_{J}(\beta_J(q=0),\bar{a}(\beta_J(q=0),q=0))\\
	+&\,U_{J}(\beta_J(q),\bar{a}(\beta_J(q),q))-U_{J}(\beta_J(q=0),\bar{a}(\beta_J(q=0),q)).
\end{align*}

The sum \(\Delta_2 + \Delta_3\) captures the total effect of the stakeholder's private information. This aggregated effect is strictly positive for any informative signal \(s\) and increases in \(q\). The key question is whether the magnitude of \(\Delta_2 + \Delta_3\) can exceed that of \(\Delta_1\). Part 3 of Lemma \ref{Lemma: ExtremeSignals} addresses this question, demonstrating that when the signal is perfectly informative (\(q=1\)), we indeed have \(\Delta_2 + \Delta_3 > |\Delta_1|\). Therefore, the two extreme cases (Parts 1 and 3), combined with the monotonicity of \(U_J\) with respect to \(q\) (Part 2 of Lemma \ref{Lemma: ExtremeSignals}), fully explain the main trade-off established in Proposition \ref{Prop:trade-offStakeinq}.

\section{Comparison of Audit Risk}
\label{Section: AuditorandRisk}
We define audit risk as the probability of an undetected audit failure. The audit risk in the two regimes can be expressed as 
\begin{align}
		AuditRisk_J&=\Pr(\theta=B)(1-\beta_J)(1-\bar{a}),\label{riskJ}\\
		AuditRisk_C&=\Pr(\theta=B)(1-\beta_C)(1-a_C).\label{riskC}
\end{align}
where $\beta_J=\frac{\Pr(\theta=B)L}{k}\bar{a}$ and $\beta_C=\frac{\Pr(\theta=B)L}{k}a_C$, following equations \eqref{betaJ} and \eqref{expressionbetaC}. The next proposition compares $AuditRisk_J$ with $AuditRisk_C$. 

\begin{proposition} \label{Prop: mainAuditor}
There exists a threshold $q^\ddagger\in(0,1)$ such that the audit risk is higher in Regime J than in Regime C if and only if $q<q^\ddagger$. 
\end{proposition}

Note that since $AuditRisk_J$ and $AuditRisk_C$ take the same form, their comparison reduces to a comparison between the critical verification level in Regime J, $\bar{a}$, and the verification level in Regime C, $a_C$. To illustrate the idea of Proposition \ref{Prop: mainAuditor}, the following lemma compares $\bar{a}$ and $a_C$ at the extremes of signal informativeness.

\begin{lemma} \label{Lemma: compareofa} 
The committed verification level in Regime~C is higher than the critical verification level in Regime~J when the signal~$s$ is uninformative, and lower than the latter when $s$ is perfectly informative. That is, $\bar{a}(\beta_J(q=0),q=0)<a_C<\bar{a}(\beta_J(q=1),q=1)$.
\end{lemma}

The intuition for Lemma \ref{Lemma: compareofa} follows from previous results. As shown in the proof of Lemma \ref{Lemma: ExtremeSignals}, the stakeholder chooses a higher level of critical verification in Regime C than in Regime J when the signal $s$ is uninformative, as the stakeholder can better incentivize auditor effort through committing to a verification level. As the signal $s$ becomes more informative, the critical verification level in Regime J increases, as shown in Part 3 of Proposition \ref{Prop: beta}. Lemma \ref{Lemma: compareofa} further establishes that when $s$ is perfectly informative, the critical verification level in Regime J exceeds that in Regime C.

\section{Comparative Statics and Empirical Implications}
\label{Section: Statics}
In this section, we explore how audit incentive ($l\equiv L/k$) and verification incentive ($i\equiv I/c$) affect the comparison of the two regimes.  To ensure tractability, we adopt a linear signal structure in Regime J. Under the linear information structure, the stakeholder in Regime J is either perfectly informed about the underlying state with probability $q$ or receives nothing with probability $1-q$. In the proof of Proposition \ref{Prop: ComparativeStaticTradeoffStake}, we establish that a higher $q$ implies higher informativeness in the Blackwell sense. Thus, all our main results apply to this signal structure. 

The next proposition shows our main comparative statics analysis. 
\begin{proposition}\label{Prop: ComparativeStaticTradeoffStake}
Both $q^\dagger$ and $q^{\ddagger}$ decrease in the verification incentive and increase in the audit incentive. That is, $\frac{\partial q^\dagger}{\partial i}<0$, $\frac{\partial q^\dagger}{\partial l}>0$, $\frac{\partial q^\ddagger}{\partial i}<0$, and $\frac{\partial q^\ddagger}{\partial l}>0$.
\end{proposition}

Proposition \ref{Prop: ComparativeStaticTradeoffStake} offers two insights. First, a decrease in verification incentive, captured by $\frac{I}{c}$, makes Regime C more attractive relative to Regime J. The intuition behind this asymmetric impact lies in the timing of the verification decision and its effect on auditor deterrence. When the verification incentive weakens (e.g., due to higher verification costs), the resource allocation effect of private information becomes weaker, and the stakeholder reduces verification efforts. In the Judgment Regime, this decision occurs ex-post; the stakeholder cannot credibly threaten to verify a suspected failure if the immediate costs outweigh the benefits. Anticipating this ex-post reluctance, the auditor reduces her effort, which also weakens the deterrence value of the stakeholder's private information. In contrast, Regime C avoids this issue because it allows the stakeholder to choose verification levels ex-ante. Consequently, while weaker verification incentives harm both regimes, the ex-ante commitment mechanism partially insulates Regime C by preserving its deterrence power. 

Second, an increase in audit incentive, captured by $\frac{L}{k}$, makes Regime C more attractive relative to Regime J. Intuitively, both the deterrence power of commitment and private information become stronger when it is easier for the stakeholder to motivate the auditor to exert effort in auditing. However, in Regime J, the deterrence effect of private information is partially offset by the stakeholder's anticipation of higher auditor effort and the corresponding reduction in the verification level. Thus, while both mechanisms become more valuable when the audit incentive increases, the value of commitment rises more, reducing the relative appeal of Regime J.

Figures \ref{fig:qdaggeri} and \ref{fig:qdaggerl} illustrate the above two insights for $q^\dagger$.\footnote{The figures are schematic and not drawn to scale.} These upward-sloping lines, $G(\cdot)$, represent the relative value or relative advantage of private information (Regime J) over the commitment (Regime C) from the stakeholder's perspective, before and after the change in the relevant incentive. Their vertical intercepts, $\mathit{Deter}(\cdot)$, reflect the relative value of commitment, and their crossing points with the $q$-axis mark $q^\dagger$. As shown in Figure \ref{fig:qdaggeri}, an increase in verification incentive from $i$ to $i'$ increases both the absolute value of the intercept and the slope, which causes the line to shift downward and rotate counterclockwise, making $q^\dagger$ smaller. Conversely, an increase in audit incentive from $l$ to $l'$ causes the value of commitment to rise more than that of private information, shifting $q^\dagger$ to the right (Figure \ref{fig:qdaggerl}). 
\begin{figure}[H]
\centering
\caption{\textbf{Illustration of $\partial q^\dagger/\partial i<0$}}
\label{fig:qdaggeri}
    \includegraphics[width=1\linewidth]{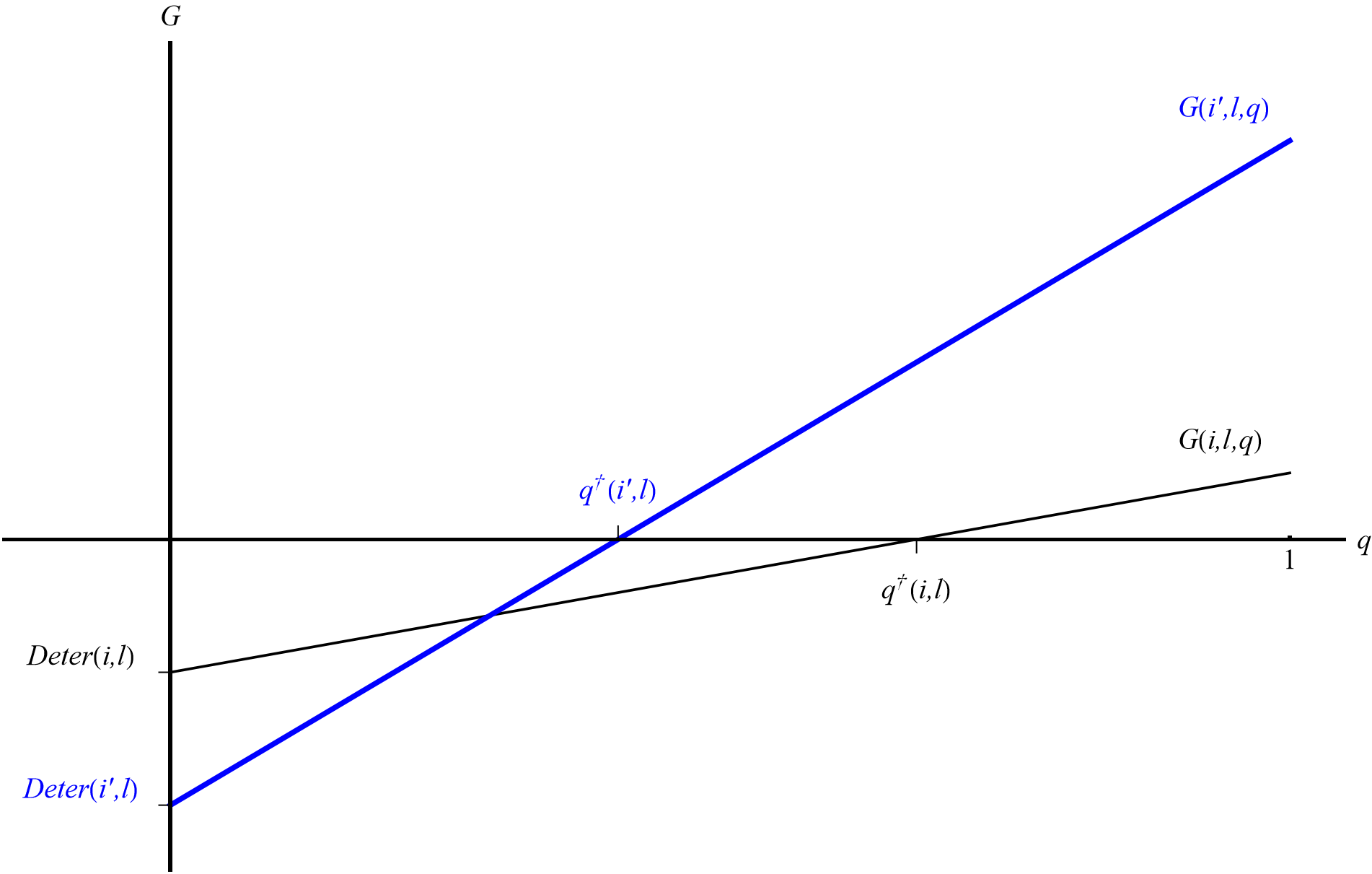}
\end{figure}

\begin{figure}[H]
\centering
\caption{\textbf{Illustration of $\partial q^\dagger/\partial l>0$}}
\label{fig:qdaggerl}
    \includegraphics[width=1\linewidth]{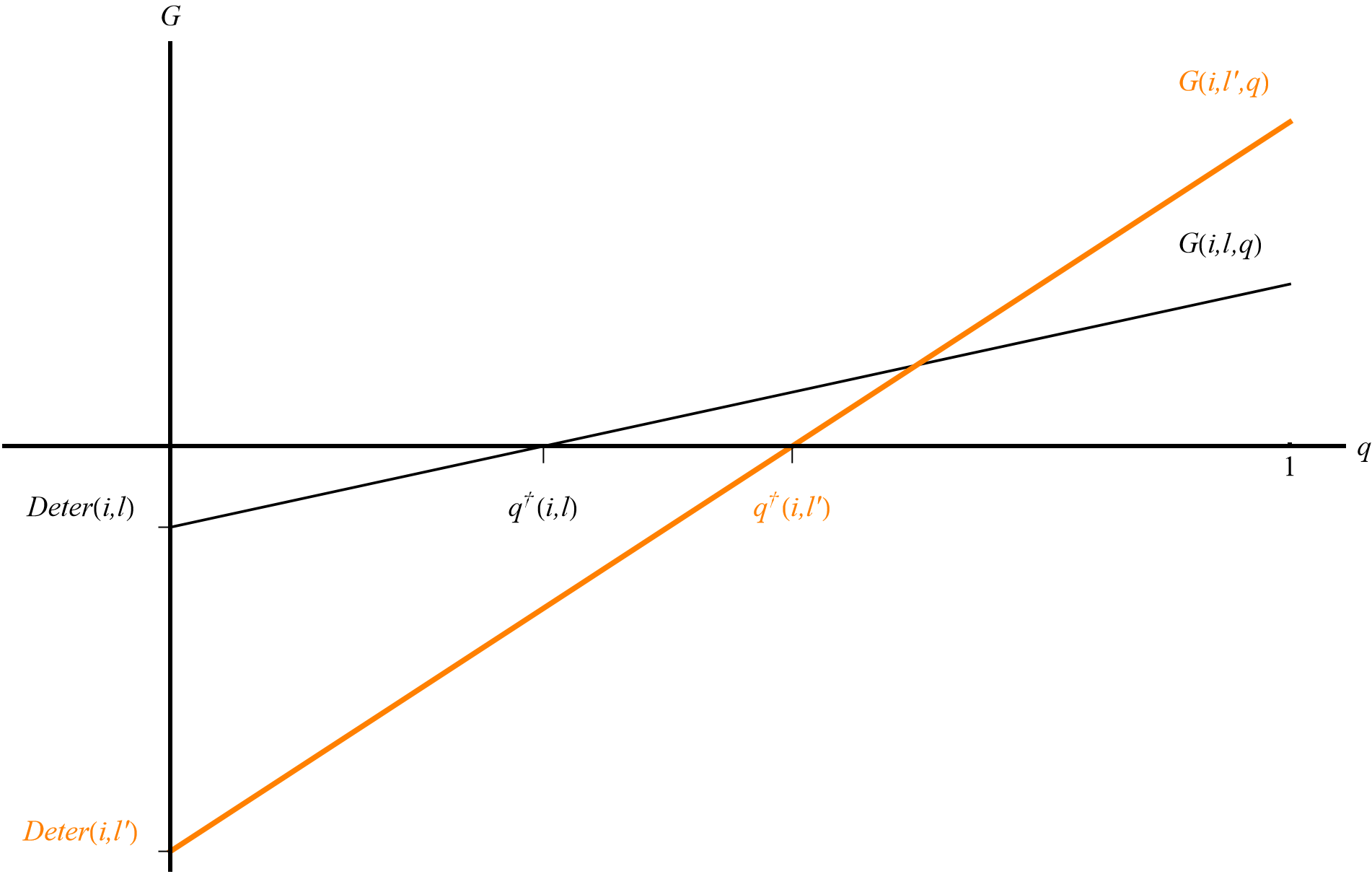}
\end{figure}

\subsection*{Empirical Implications}
\label{Section: EmpiricalImplications}
As \citet{lohlein2016peer} note, research confirms that both peer review and PCAOB inspection can enhance audit quality, though findings on how these approaches are recognized in decision-making remain mixed. Recent studies show that PCAOB inspections have improved reporting credibility and accountability through mechanisms such as stronger enforcement, greater public oversight, and partner identification requirements \citep{defond2010should, burke2019audit, gipper2020public}. However, these effects are confined to specific dimensions of audit quality and vary across inspection frequency, outcomes, firm type, and client characteristics \citep{gunny2013pcaob, lamoreaux2016does, fung2017does, aobdia2018impact, cunningham2019s}. For instance, \citet{khurana2021pcaob} find that PCAOB inspections primarily improve audit quality among Big 4 firms. Other studies suggest that inspections may reshape the audit market by disproportionately burdening smaller firms, reducing their ability to accept or retain public clients and strengthening their incentives to exit the market \citep{daugherty2010pcaob, defond2011effect, aobdia2017regulatory}.

Our model offers a potential explanation for this heterogeneity: the effectiveness of PCAOB inspection (Regime C) relative to peer review (Regime J) depends critically on the underlying verification and audit incentives, which likely differ across firm size and client complexity. When verification costs are high or audit incentives are strong, conditions that are more prevalent in complex audit engagements, our comparative statics predict that the PCAOB inspection commitment mechanism should dominate. This prediction is consistent with the stronger effects documented for larger and more complex audits.

The model also predicts opposite relations between auditor incentives and verification intensity across the two regimes. Under peer review, stronger auditor incentives reduce ex-post verification:
$
\frac{\partial \bar{a}}{\partial l}<0.
$
 Under PCAOB inspection, stronger auditor incentives increase committed verification:
$
\frac{\partial a_C}{\partial l}>0.
$ Thus, ceteris paribus, verification intensity should be negatively associated with auditor incentives under peer review but positively associated with auditor incentives under PCAOB inspection. 

As shown in Section \ref{Section: AuditorandRisk}, audit risk is determined by the verification level in each regime. Therefore, the above predictions relate to a vast empirical literature examining how audit incentives, regulatory oversight, and litigation risk affect audit risk and audit quality \citep[e.g.,][]{aobdia2018impact, christensen2024costs, johnson2019auditors, westermann2019pcaob, chen2025pcaob}. Our model further predicts that the association between audit incentives and verification intensity differs systematically between peer review and PCAOB inspection.

\section{Conclusions}
\label{Section: Conclusions}
This paper develops a theoretical framework to evaluate the fundamental trade-off between the peer review system (Judgment Regime) and government inspection (Commitment Regime). We demonstrate that the Judgment Regime benefits from ex-post private information, which improves resource allocation and deters audit failures. Conversely, the Commitment Regime relies on the deterrence power of ex-ante commitment. We find that the stakeholder prefers the Judgment Regime only when the peer reviewer's private signal is sufficiently informative; otherwise, the Commitment Regime dominates. Furthermore, our comparative statics show that higher verification costs or stronger audit incentives make the Commitment Regime relatively more attractive. Our findings provide a theoretical foundation for the mixed empirical evidence regarding the shift from AICPA peer reviews to PCAOB inspections, suggesting that the optimal regime depends heavily on the auditor's incentive, the stakeholder's verification incentive, and the strength of the stakeholder's private information. 

Future research could expand upon our framework by exploring other potential drawbacks of the peer review system. For example, while our current model highlights the inability to commit to a level of effort as the key disadvantage of peer review, future work could incorporate the reciprocal nature of these relationships, which can severely compromise auditor independence. In addition, future studies could investigate how differences in auditors' industry specializations and expertise might lead to the adoption of varying standards, resulting in a lack of consistency. Finally, future studies could examine the consequences of absent external oversight, specifically how weak enforcement of audit standards arises when peer reviewers lack the necessary authority or motivation to effectively police non-compliance. 

\newpage
\appendix
\section{Proofs}
\subsection{Proof of Proposition \ref{Prop: equilibriumJ}}
We have derived the verification level and the audit effort in equations \eqref{aw} and \eqref{equilibriumJ}, which are determined by the root of $F_J(\beta)$ in \eqref{FJbeta}. To prove the existence of a unique equilibrium in Regime J, we verify the three properties of $F_J(\beta)$ stated in the main text. Plugging $\Pr(\theta=B)=1/2$ into $F_J(\beta)$ gives
\begin{align*}
    F_J(\beta)=\beta-\frac{LI}{2kc}\sum_t \frac{(1-\beta)[\Pr(s=\omega_t|\theta=B)]^2}{(1-\beta)\Pr(s=\omega_t|\theta=B)+\Pr(s=\omega_t|\theta=G)}.
\end{align*}

Evaluated at $\beta=0$ and $\beta\to1^-$, we have
\begin{align*}
    F_J(0)&=-\frac{LI}{2kc}\sum_t \frac{[\Pr(s=\omega_t|\theta=B)]^2}{\Pr(s=\omega_t|\theta=B)+\Pr(s=\omega_t|\theta=G)}<0,\\
    \lim_{\beta\to1^-}F_{J}(\beta)&=1-\frac{LI}{2kc}\sum_{t:\Pr(s=\omega_t|\theta=G)=0}\Pr(s=\omega_t|\theta=B)\geq1-\frac{LI}{2kc}>0.
\end{align*}
Moreover,
\begin{align*}
F_J'(\beta) = 1 + \frac{LI}{2kc} \sum_t \frac{[\Pr(s=\omega_t|\theta=B)]^2 \Pr(s=\omega_t|\theta=G)}{[(1-\beta)\Pr(s=\omega_t|\theta=B) + \Pr(s=\omega_t|\theta=G)]^2}>0,
\end{align*}
which implies the increasing monotonicity of $F_J(\beta)$. Collecting $F_J(0)<0$, $\lim_{\beta\to1^-}F_J(\beta)>0$, $F^{'}_J(\beta)>0$, and the continuity of $F_J(\beta)$, there exists a unique solution for $F_J(\beta)=0$ by the Intermediate Value Theorem. This proves Proposition \ref{Prop: equilibriumJ}.

\subsection{Proof of Proposition \ref{Prop: beta}}
Consider any two signals $s_1$ and $s_2$ in the following proof, such that $s_1$ and $s_2$ are independent of $r$ conditional on $\theta$, and $s_1$ Blackwell dominates $s_2$.

\textbf{Proof of Part 1}

Based on equation \eqref{bara} and expression \eqref{FJbeta},
\begin{align*}
\bar{a}&=\frac{I}{c}\sum_t \frac{(1-\beta_J)[\Pr(s=\omega_t|\theta=B)]^2}{(1-\beta_J)\Pr(s=\omega_t|\theta=B)+\Pr(s=\omega_t|\theta=G)}\\
&=\frac{2I(1-\beta_J)}{c}\sum_t \frac{\Pr(s=\omega_t)\Pr(\theta=B|s=\omega_t)^2}{(1-\beta_J)\Pr(\theta=B|s=\omega_t)+\Pr(\theta=G|s=\omega_t)}\\
&=\frac{2I(1-\beta_J)}{c}\sum_t \frac{\Pr(s=\omega_t)\Pr(\theta=B|s=\omega_t)^2}{1-\beta_J\Pr(\theta=B|s=\omega_t)}\\
&=\frac{2I(1-\beta_J)}{c}E\left[\frac{\Pr(\theta=B|s)^2}{1-\beta_J\Pr(\theta=B|s)}\right].
\end{align*}
The second equality applies Bayes' rule to calculate $\Pr(\theta=B|s=\omega_t)$ and $\Pr(\theta=G|s=\omega_t)$. The third equality uses $\Pr(\theta=G|s=\omega_t)=1-\Pr(\theta=B|s=\omega_t)$. The last equation rewrites the result in expectation form.

Since $s_1$ Blackwell dominates $s_2$ (i.e., $q_1>q_2$), the posterior induced by $s_1$ is a mean-preserving spread of the posterior induced by $s_2$. By the definition of mean-preserving spread, $E\left[\frac{\Pr(\theta=B|s_1)^2}{1-\beta_J\Pr(\theta=B|s_1)}\right]>E\left[\frac{\Pr(\theta=B|s_2)^2}{1-\beta_J\Pr(\theta=B|s_2)}\right]$ holds because the function is strictly convex with respect to the posterior belief. Therefore, $\bar{a}(\beta_J,q_1)>\bar{a}(\beta_J,q_2)$, which implies that $\bar{a}(\beta_J,q)$ is increasing in $q$. This proves Part 1.

\textbf{Proof of Part 2}

Based on \eqref{aw}, it is obvious that $a(\omega_t)$ decreases in $\beta_J$, as $\Pr(\theta=B|r=g,s=\omega_t)$ decreases in $\beta_J$. According to \eqref{bara}, $\bar{a}(\beta_J,q)$ decreases in $\beta_J$ as well. 

We now prove Part 2 through contradiction. Suppose  $\beta_J(q_1)\le\beta_J(q_2)$. By equation \eqref{betaJ}, we have 
\begin{align*}
	\beta_J(q_1)&=\frac{L}{2k}\bar{a}(\beta_J(q_1),q_1)\\
	&\ge\frac{L}{2k}\bar{a}(\beta_J(q_2),q_1)\\
	&> \frac{L}{2k}\bar{a}(\beta_J(q_2),q_2)\\
	&=\beta_J(q_2).
\end{align*}
The first inequality holds as $\beta_J(q_1)\le\beta_J(q_2)$ and $\bar{a}(\beta_J,q)$ decreases in $\beta_J$. The second inequality follows from Part 1 of Proposition \ref{Prop: beta}. The last line is a contradiction to $\beta_J(q_1)\le\beta_J(q_2)$. Thus, we must have $\beta_J(q_1)>\beta_J(q_2)$. This proves Part 2. 

\textbf{Proof of Part 3}
\begin{align*}
\bar{a}(\beta_J(q_1),q_1)&=\frac{2k}{L}\beta_J(q_1)\\
                        &>\frac{2k}{L}\beta_J(q_2)\\
                        &=\bar{a}(\beta_J(q_2),q_2).
\end{align*}
The first and the third lines hold due to equation \eqref{betaJ}. The second line follows from Part 2 of the Proposition. Thus, $\bar{a}(\beta_J(q),q)$ increases in $q$. This proves the last part of the Proposition.

\subsection{Proof of Corollary \ref{coroJ}}
We first examine how audit and verification incentives affect $\beta_J$. 
Applying the Implicit Function Theorem to equation \eqref{equilibriumJ} yields
\begin{align}
    \frac{\partial \beta_J}{\partial l}=\frac{2\beta_J}{l\left(2 + i l \sum_t \frac{[\Pr(s=\omega_t|\theta=B)]^2 \Pr(s=\omega_t|\theta=G)}{[(1-\beta_J)\Pr(s=\omega_t|\theta=B) + \Pr(s=\omega_t|\theta=G)]^2} \right)} > 0,\label{betaJonl}\\
    \frac{\partial \beta_J}{\partial i} = \frac{2\beta_J}{i\left(2 + i l \sum_t \frac{[\Pr(s=\omega_t|\theta=B)]^2 \Pr(s=\omega_t|\theta=G)}{[(1-\beta_J)\Pr(s=\omega_t|\theta=B) + \Pr(s=\omega_t|\theta=G)]^2} \right)} > 0.\label{betaJoni}
\end{align}
The two inequalities in equations \eqref{betaJonl} and \eqref{betaJoni} hold due to positive signs of $\beta_J$ and other components. 

We next examine how audit and verification incentives affect $\bar{a}$. Recall from equation \eqref{betaJ} that $\bar{a}=\frac{2\beta_J}{l}$. Therefore,
\begin{align}
\label{baraoni}
    \frac{\partial \bar{a}}{\partial i}=\frac{2}{l} \frac{\partial \beta_J}{\partial i}>0.
\end{align}
This inequality follows from $\frac{\partial \beta_J}{\partial i}>0$, which has just been proved in \eqref{betaJoni}.

Similarly, taking the derivative of equation \eqref{betaJ} with respect to $l$ yields
\begin{align}
\label{baraonl}
    \frac{\partial \bar{a}}{\partial l}=\frac{2l\frac{\partial \beta_J}{\partial l}-2\beta_J}{l^2} = -\frac{2\beta_J}{l^2} \left( \frac{i l \sum_t \frac{[\Pr(s=\omega_t|\theta=B)]^2 \cdot \Pr(s=\omega_t|\theta=G)}{\left[(1-\beta_J)\Pr(s=\omega_t|\theta=B) + \Pr(s=\omega_t|\theta=G)\right]^2}}{2 + i l \sum_t \frac{[\Pr(s=\omega_t|\theta=B)]^2 \cdot \Pr(s=\omega_t|\theta=G)}{\left[(1-\beta_J)\Pr(s=\omega_t|\theta=B) + \Pr(s=\omega_t|\theta=G)\right]^2}}\right)<0.
\end{align}
The second equality uses equation \eqref{betaJonl} and the inequality holds obviously. Thus, we have proved Corollary \ref{coroJ}.

\subsection{Proof of Proposition \ref{Prop: equilibriumC}}
The equilibrium verification level in Regime C is defined in equation \eqref{equilibriumC} in the main text, which is reproduced here
\begin{align*}
    3l a^{2} - (8 + 4il)a + 2i(l + 2) = 0.
\end{align*}

Denoting the LHS of this equation as $F_C(a)$, we can easily verify that $F_C(0)=2i(l+2)>0$, $F_C(1)=3l-2il+4i-8<0$, and $F_C^{'}(a)=6al-(8+4il)<0$. Thus,  by the Intermediate Value Theorem and the continuity of $F_C(a)$, there exists a unique solution $a_C\in(0,1)$ such that $F_C(a_C)=0$. This proves the existence and uniqueness of the equilibrium in Regime C.

\subsection{Proof of Corollary \ref{coroC}}
Applying the Implicit Function Theorem to equation \eqref{equilibriumC} yields
\begin{align}
    \frac{\partial a_C}{\partial i}&=\frac{l+2-2la_C}{2il+4-3la_C}>0,\label{aConi}\\
    \frac{\partial a_C}{\partial l}&=\frac{3a_C^2-4ia_C+2i}{8+4il-6la_C}=\frac{3(a_C-\frac{2}{3}i)^2+i(2-\frac{4}{3}i)}{8+4il-6la_C}>0.\label{aConl}
\end{align}
The inequality in \eqref{aConi} holds since both the denominator and the numerator are positive. The second equality in \eqref{aConl} rewrites the result of $\frac{\partial a_C}{\partial l}$. These two show the positive effects of verification and audit incentives on $a_C$.

Differentiating \eqref{expressionbetaC} with respect to incentives $i$ and $l$ leads to 
\begin{align*}
    \frac{\partial \beta_C}{\partial i}&=\frac{l}{2}\frac{\partial a_C}{\partial i},\\
    \frac{\partial \beta_C}{\partial l}&=\frac{a_C+l\frac{\partial a_C}{\partial l}}{2}.
\end{align*}
Both $\frac{\partial \beta_C}{\partial i}>0$ and $\frac{\partial \beta_C}{\partial l}>0$ follow from $\frac{\partial a_C}{\partial i}>0$ and $\frac{\partial a_C}{\partial l}>0$. Therefore, this proves Corollary \ref{coroC}.

\subsection{Proof of Lemma \ref{Lemma: ExtremeSignals}}
Our proof proceeds in three steps. First, we show that the stakeholder's expected utility in Regime J can be expressed as a function of the audit effort $\beta_J(q)$ and the critical verification level $\bar{a}(\beta_J(q),q)$. Then, we compare the equilibrium decisions in Regime C with decisions in Regime J with an uninformative signal (i.e., $q=0$). Finally, we prove three parts of the Lemma separately.

\subsubsection*{Step 1: Show that $U_J$ is a function of $\beta_J(q)$ and $\bar{a}(\beta_J(q),q)$}

\begin{align*}
	&U_{J}=\sum_t \Pr(r=g,s=\omega_t)(\Pr(\theta=G|r=g,s=\omega_t)(X-I)\\
    &\quad+(1-a(\omega_t))\Pr(\theta=B|r=g,s=\omega_t)(-I) -\frac{1}{2}ca(\omega_t)^2)\\
	&=\sum_t (\Pr(\theta=G,r=g,s=\omega_t)(X-I)-(1-a(\omega_t))\Pr(\theta=B,r=g,s=\omega_t)I\\
	&\quad -\frac{\Pr(r=g,s=\omega_t)ca(\omega_t)^2}{2})\\
    &=\Pr(\theta=G)(X-I)-I\Pr(\theta=B)\sum_t\Pr(r=g,s=\omega_t|\theta=B)(1-a(\omega_t))\\
    &\quad -\frac{\sum_t\Pr(r=g,s=\omega_t)ca(\omega_t)^2}{2}\\
    &=\Pr(\theta=G)(X-I)-I(1-\beta_J)\Pr(\theta=B)(1-\sum_t \Pr(s=\omega_t|\theta=B)a(\omega_t))\\
    &\quad-I\sum_t\Pr(\theta=B,r=g,s=\omega_t)\frac{a(\omega_t)}{2}\\
    &=\Pr(\theta=G)(X-I)-I\Pr(\theta=B)(1-\beta_J)(1-\bar{a})-I\Pr(\theta=B)(1-\beta_J)\frac{\bar{a}}{2}\\
    &=\Pr(\theta=G)(X-I)-\Pr(\theta=B)(1-\beta_J(q))(1-\frac{\bar{a}(\beta_J(q),q)}{2})I.
\end{align*}
The second equality uses Bayes' rule to calculate $\Pr(\theta=G|r=g,s=\omega_t)$ and $\Pr(\theta=B|r=g,s=\omega_t)$. The third equality uses $\sum_t\Pr(\theta=G,r=g,s=\omega_t)=\Pr(\theta=G)$ and $\Pr(\theta=B,r=g,s=\omega_t)=\Pr(\theta=B)\Pr(r=g,s=\omega_t|\theta=B)$. The fourth equality uses equation \eqref{aw} and the conditional independence of $r$ and $s$. The fifth equality rewrites $\Pr(\theta=B,r=g,s=\omega_t)$ again and uses equation \eqref{bara} that $\bar{a}=\sum_t\Pr(s=\omega_t|\theta=B,r=g)a(\omega_t)=\sum_t\Pr(s=\omega_t|\theta=B)a(\omega_t)$. The last line shows that the stakeholder's expected utility can be expressed as a function of $\beta_J(q)$ and $\bar{a}(\beta_J(q),q)$. That is, 
\begin{equation}
\label{SimplifiedUJ}
    U_J(\beta_J(q),\bar{a}(\beta_J(q),q))=\frac{1}{2}(X-I)-\frac{1}{2}(1-\beta_J(q))(1-\frac{\bar{a}(\beta_J(q),q)}{2})I.
\end{equation}

\subsubsection*{Step 2: Compare $(\beta_J(q=0),\bar{a}(\beta_J(q=0),q=0))$ with $(\beta_C,a_C)$}
Solving the equilibrium in Regime J with an uninformative signal $s$ (i.e., equations \eqref{aw}, \eqref{bara} and \eqref{equilibriumJ} at $q=0$), we have
\begin{align}
\bar{a}(\beta_J(q=0),q=0)&=\frac{4 + il - \sqrt{16 + i^2l^2}}{2l},\label{Exp: ag}\\
\beta_J(q=0)&=\frac{4 + il - \sqrt{16 + i^2l^2}}{4}.\label{Exp: betaJempty}
\end{align}

Solving equations \eqref{expressionbetaC} \eqref{equilibriumC} leads to the equilibrium in Regime C
\begin{align}
    a_C&=\frac{4 + 2i l - \sqrt{16 + 4il + 2il^2( 2i-3)}}{3l},\label{Exp: aC}\\
    \beta_C&=\frac{4 + 2i l - \sqrt{16 + 4il + 2il^2( 2i-3)}}{6}.\label{Exp: betaC}
\end{align}

Let $n\equiv\sqrt{16 + i^2l^2}$ and $m\equiv\sqrt{16 + 4il + 2il^2(2i-3)}$. It can be easily verified that $3n>2m$ since $(3n-2m)(3n+2m)=80-16il+24il^2-7i^2l^2>0$. Comparing $\bar{a}(\beta_J(q=0),q=0)$ and $a_C$ yields
\begin{align*}
   &a_C-\bar{a}(\beta_J(q=0),q=0)\\
    =\quad &\frac{3n-2m-(4-il)}{6l}\\
    =\quad &\frac{2 \left(16 + 2 i l (1 - (1 - i) l) - m n\right)}{l\left(3n-2m+(4-il)\right)}\\
    =\quad &\frac{4\left(6i l^2(2-i)+4i l^2(1-i l) + i^3 l^3(2-l) + 2 i^2 l^4\right) }{l\left(3n-2m+(4-il)\right)\left(16 + 2 i l (1 - (1 - i) l) + m n\right)}>0.
\end{align*}
These lines use basic rules of algebra. The inequality in the last line holds as both the numerator and denominator are positive ($0 < i < 1$, $0 < l < 1$, and $3n > 2m$). This proves $a_C>\bar{a}(\beta_J(q=0),q=0)$.

Then, $\beta_C>\beta_J(q=0)$ follows from equations \eqref{betaJ}, \eqref{expressionbetaC}, and $a_C>\bar{a}(\beta_J(q=0),q=0)$.

\subsubsection*{Step 3: Prove three parts of the Lemma}
\textbf{Proof of Part 1: $U_{J}(q=0)<U_{C}$}

The stakeholder's ex-ante utility in Regime J with an uninformative signal can be computed as
\begin{align*}
		U_{J}(q=0)& =\Pr(\theta=G)(X-I)-\Pr(\theta=B)\left(1-\beta_J(q=0)\right)(1-\bar{a}(q=0))I\\
        &\quad -\left(\Pr(\theta=G)+\Pr(\theta=B)(1-\beta_J(q=0))\right)\frac{1}{2}c \bar{a}(q=0)^2.
\end{align*}

Recall that the stakeholder's utility in Regime C can be computed as
\begin{align*}
	U_{C}& =\Pr(\theta=G)(X-I)-\Pr(\theta=B)\left(1-\beta_C\right)(1-a_C)I\\
    &\quad -\left(\Pr(\theta=G)+\Pr(\theta=B)(1-\beta_C)\right)\frac{1}{2}c a_C^2.
\end{align*}

Since equations \eqref{betaJ} and \eqref{expressionbetaC} indicate that $\bar{a}(q=0)=\frac{\beta_J(q=0)}{l\Pr(\theta=B)}$ and $a_C=\frac{\beta_C}{l\Pr(\theta=B)}$, the stakeholder in these two cases shares the same utility form, which can be written as
\begin{align*}
	U_{R_0}& =\Pr(\theta=G)(X-I)-\Pr(\theta=B)\left(1-\beta_{R_0}\right)(1-a_{R_0})I\\
    &\quad -\left(\Pr(\theta=G)+\Pr(\theta=B)(1-\beta_{R_0})\right)\frac{1}{2}c a_{R_0}^2,
\end{align*}
where $R_0\in\{J_0,C\}$ denotes the Regime J with $q=0$ or Regime C.

Plugging $\beta_{R_0}=\Pr(\theta=B) l a_{R_0}$ into $U_{R_0}$, and differentiating with respect  to $a_{R_0}$ leads to
\begin{align} \label{EquiaR0}
\frac{3}{2}l\Pr(\theta=B)^2 a_{R_0}^2-\left(1+2il\Pr(\theta=B)^2\right)a_{R_0}+
i\Pr(\theta=B)(1+l\Pr(\theta=B))=0.
\end{align}
Since equation \eqref{EquiaR0} is essentially the same as \eqref{equilibriumC}, the unique solution for this equation equals $a_C$ and $U_{R_0}$ increases in $a_{R_0}$ over the interval $(0,a_C]$. Thus, $U_{J}(q=0)<U_{C}$ follows from $\bar{a}(\beta_J(q=0),q=0)<a_C$ in Step 2.

\textbf{Proof of Part 2: $\frac{dU_J(\beta_J(q),\bar{a}(\beta_J(q),q))}{dq}>0$}

Recall from Step 1 that the stakeholder's ex-ante utility in Regime J is
\begin{align*}
    U_J(\beta_J(q),\bar{a}(\beta_J(q),q))=\Pr(\theta=G)(X-I)-\Pr(\theta=B)(1-\beta_J(q))(1-\frac{\bar{a}(\beta_J(q),q)}{2})I.
\end{align*}
Therefore, $\frac{dU_J(\beta_J(q),\bar{a}(\beta_J(q),q))}{dq}>0$ follows from $\frac{d\beta_J(q)}{dq}>0$ and $\frac{d\bar{a}(\beta_J(q),q)}{dq}>0$ (Parts 2 and 3 of Proposition \ref{Prop: beta}).

\textbf{Proof of Part 3: $U_{J}(q=1)>U_{C}$}

Recall that the stakeholder's ex-ante utility with a perfectly informative signal is
\begin{align*}
	U_{J}(q=1)=\frac{1}{2}(X-I)-\frac{1}{2}\left(1-\beta_J(q=1)\right)(1-\frac{\bar{a}(q=1)}{2})I,
\end{align*}
where 
\begin{align}
\bar{a}(\beta_J(q=1),q=1)&= i, \label{baraq1}\\
\beta_J(q=1)&=\frac{i l}{2}. \label{betaJq1}
\end{align}
Rearranging and simplifying $U_{J}(q=1)-U_{C}$, we have
\begin{align*}
	&U_{J}(q=1)-U_{C} \quad \propto \quad h-2 \cdot m^3,
\end{align*}
where $h\equiv\left( -11 i^3 l^3 + 18 i^2 l^3 + 78 i^2 l^2 - 72 i l^2 + 48 i l + 128 \right)$, and $m\equiv\sqrt{16 + 4il + 2il^2(2i-3)}$, both positive for $0<i,l<1$. 

We proceed by showing that the minimum of $h-2m^3$ is positive. Let $u\equiv4il-6il^2+4i^2l^2$, so that $m=\sqrt{16+u}$. It is easily known that $2m^3=2m(16+u)\ \leq 128+12u+\frac{u^2}{4} $ since $8m\leq 16+m^2$. Thus, rearranging and simplifying $h-2m^3$ gives
\begin{align*}
h-2m^3&\geq h-128-12u-\frac{u^2}{4}\\
&=26i^2l^2+30i^2l^3-19i^3l^3-9i^2l^4+12i^3l^4-4i^4l^4\\
&=i^2l^2(26-19il)+i^2l^3(30-9l)+i^3l^4(12-4i)>0.
\end{align*}
This inequality holds obviously for $0<i,l<1$. It shows $h-2m^3>0$ and therefore proves $U_{J}(q=1)>U_{C}$.

\subsection{Proof of Proposition \ref{Prop:trade-offStakeinq}}
The proof proceeds in two steps. First, we show some important properties of $\Delta_1$, $\Delta_2$, and $\Delta_3$ in the main text. Second, we show that $q^\dagger$ exists.

\subsubsection*{Step 1: Important properties of $\Delta_1$, $\Delta_2$, and $\Delta_3$}
$\Delta_1$ in \eqref{Eqn:commit} is defined as $U_{J}(\beta_J(q=0),\bar{a}(\beta_J(q=0),q=0))-U_{C}(\beta_C,a_C)$, which is negative because of Part 1 of Lemma \ref{Lemma: ExtremeSignals}.

Recall from the proof of Lemma \ref{Lemma: ExtremeSignals} that the stakeholder's utility is
\begin{align*}
    U_J(\beta_J(q),\bar{a}(\beta_J(q),q))=\Pr(\theta=G)(X-I)-\Pr(\theta=B)(1-\beta_J(q))(1-\frac{\bar{a}(\beta_J(q),q)}{2})I.
\end{align*}
Therefore, $\Delta_2$ defined in \eqref{Eqn:resource} can be written as 
\begin{align*}
\Delta_2\coloneqq & \quad U_{J}(\beta_J(q=0),\bar{a}(\beta_J(q=0),q))-U_{J}(\beta_J(q=0),\bar{a}(\beta_J(q=0),q=0))\\
=& \quad \frac{I}{2}\Pr(\theta=B)(1-\beta_J(q=0))\left(\bar{a}(\beta_J(q=0),q)-\bar{a}(\beta_J(q=0),q=0)\right).
\end{align*}
Since $\bar{a}(\beta_J(q=0),q)>\bar{a}(\beta_J(q=0),q=0)$ (Part 1 of Proposition \ref{Prop: beta}), we have $\Delta_2>0$.

Similarly, $\Delta_3$ defined in \eqref{Eqn:deter} is calculated as 
\begin{align*}
    \Delta_3\coloneqq & \quad U_{J}(\beta_J(q),\bar{a}(\beta_J(q),q))-U_{J}(\beta_J(q=0),\bar{a}(\beta_J(q=0),q))\\
   =&\quad \underbrace{    U_{J}(\beta_J(q),\bar{a}(\beta_J(q=0),q))-U_{J}(\beta_J(q=0),\bar{a}(\beta_J(q=0),q))}_{\text{Direct effect}}\\
   & + \underbrace{U_{J}(\beta_J(q),\bar{a}(\beta_J(q),q))-U_{J}(\beta_J(q),\bar{a}(\beta_J(q=0),q))}_{\text{Indirect effect}}.
\end{align*}
The second equality decomposes $\Delta_3$ into two terms: the direct effect of changing audit effort from $\beta_J(q=0)$ to $\beta_J(q)$, and the indirect effect via changing critical verification level from $\bar{a}(\beta_J(q=0),q)$ to $\bar{a}(\beta_J(q),q)$. The direct effect is positive due to Part 2 of Proposition \ref{Prop: beta} that $\beta_J(q)>\beta_J(q=0)$. The indirect effect is negative since we have shown in the proof for Part 2 of Proposition \ref{Prop: beta} that $\bar{a}(\beta_J(q),q)<\bar{a}(\beta_J(q=0),q)$. 

We then show the dominance of the direct effect over the indirect effect. Recall from Step 1 of the proof for Lemma \ref{Lemma: ExtremeSignals} that the stakeholder's ex-ante expected utility can be computed as
\begin{align*}
    &U_J=\Pr(\theta=G)(X-I)-I\Pr(\theta=B)\sum_t\Pr(r=g,s=\omega_t|\theta=B)(1-a(\omega_t))\\
    &\quad -\frac{\sum_t\Pr(r=g,s=\omega_t)ca(\omega_t)^2}{2}.
\end{align*}
Note that $\Pr(r=g,s=\omega_t|\theta=B)=(1-\beta_J)\Pr(s=\omega_t|B)$, and $\Pr(r=g,s=\omega_t)=\Pr(\theta=G)\Pr(s=\omega_t|G)+\Pr(\theta=B)(1-\beta_J)\Pr(s=\omega_t|B)$. Thus, $\frac{\partial U_J}{\partial \beta_J}>0$. Since the equilibrium verification levels satisfy the first-order condition, we must have $\frac{\partial U_J}{\partial a(\omega_t)}=0$ for any $t$. By the envelope theorem, we have $\frac{d U_J(\beta_J(q),\bar{a}(\beta_J(q),q))}{d\beta_J}>0$. This proves $\Delta_3>0$.

It is straightforward that
\begin{align*}
\frac{d \Delta_2+\Delta_3}{dq} &=\frac{d U_J(\beta_J(q),\bar{a}(\beta_J(q),q))-U_J(\beta_J(q=0),\bar{a}(\beta_J(q=0),q=0))}{dq}\\
&=\frac{d U_J(\beta_J(q),\bar{a}(\beta_J(q),q))}{dq}>0.
\end{align*}
The inequality follows from Part 2 of Lemma \ref{Lemma: ExtremeSignals}.

\subsubsection*{Step 2: Existence of $q^\dagger$}
We have shown in Step 1 and Lemma \ref{Lemma: ExtremeSignals} that $U_J(q=0)-U_C<0$, $U_J(q=1)-U_C>0$, $\frac{d\left(U_J(q)-U_C\right)}{dq}>0$ and the continuity of $U_J(q)$ and $U_C$. By the Intermediate Value Theorem, there must be a unique threshold $q^\dagger\in(0,1)$ such that $U_J>U_C$ if and only if $q>q^\dagger$. Thus, this proves the Proposition.

\subsection{Proof of Lemma \ref{Lemma: compareofa}}
We will focus on $\bar{a}(\beta_J(q=1),q=1)>a_C$ because $\bar{a}(\beta_J(q=0),q=0)<a_C$ has been proved in the proof for Step 2 of Lemma \ref{Lemma: ExtremeSignals}.

Comparing $a_C$ in \eqref{Exp: aC} and $\bar{a}(\beta_J(q=1),q=1)$ in \eqref{baraq1} yields
\begin{align*}
    a_C-\bar{a}(\beta_J(q=1),q=1)&=\frac{4 + 2i l - \sqrt{16 + 4il + 2il^2( 2i-3)}}{3l}-i\\
    &=\frac{(4-il)^2 - (16 + 4il + 2il^2( 2i-3))}{3l(4-il+\sqrt{16 + 4il + 2il^2( 2i-3)})}\\
    &=\frac{-6il(2-l)-3i^2l^2}{3l(4-il+\sqrt{16 + 4il + 2il^2( 2i-3)})}<0.
\end{align*}
These three lines use basic rules of algebra. The last inequality holds since the numerator is negative and the denominator is positive. Thus, this proves $a_C<\bar{a}(\beta_J(q=1),q=1)$.

In addition, it is straightforward that
\begin{equation*}
    \beta_J(q=1)=\frac{l}{2}\bar{a}(q=1)>\frac{l}{2}a_C=\beta_C.  
\end{equation*}

In sum, we have $\bar{a}(q=0)<a_C<\bar{a}(q=1)$ and $\beta_J(q=0)<\beta_C<\beta_J(q=1)$.

\subsection{Proof of Proposition \ref{Prop: mainAuditor}}
Recall that the audit risk equations in \eqref{riskJ} and \eqref{riskC} can be expressed in the same form
\begin{equation}
    AuditRisk_R=\Pr(\theta=B)(1-\beta_R)(1-a_R),\label{riskR}
\end{equation}
where $R\in\{J,C\}$, $a_R$ is the critical verification (i.e., $\bar{a}(q)$ in Regime J and the committed verification $a_C$ in Regime C), and $\beta_R\in\{\beta_J,\beta_C\}$ is the audit effort.

Since $\beta_R$ can be regarded as a function of $a_R$ (i.e., $\beta_R=\frac{L\Pr(\theta=B)}{k}a_R$), differentiating \eqref{riskR} with respect to $a_R$ yields
\begin{align*}
\frac{d AuditRisk_R}{d  a_R} &= \frac{\partial AuditRisk_R}{\partial a_R} + \frac{\partial AuditRisk_R}{\partial \beta_R} \cdot \frac{d \beta_R}{d a_R}\\
&=-\Pr(\theta=B)(1-\beta_R)-\Pr(\theta=B)(1-a_R)\cdot\frac{d \beta_R}{d a_R}<0.
\end{align*}
The inequality holds since both the direct and indirect effects of $a_R$ on $AuditRisk_R$ are negative. Thus, we have $\frac{d AuditRisk_R}{d a_R}<0$, implying that comparing audit risks is equivalent to comparing critical verification levels.

Then, based on Lemma \ref{Lemma: compareofa} that $\bar{a}(q=0)<a_C<\bar{a}(q=1)$, Part 3 of Proposition \ref{Prop: beta} that $\frac{d\bar{a}(q)}{dq}>0$, and by the Intermediate Value Theorem, there must exist a threshold $q^\ddagger\in(0,1)$ such that $AuditRisk_J>AuditRisk_C$ if and only if $q<q^\ddagger$.

\subsection{Proof of Proposition \ref{Prop: ComparativeStaticTradeoffStake}} 
Our proof proceeds in four steps. First, we establish the equivalence between Blackwell dominance (informativeness) and the probability of being informed ($q$) under the linear information structure. Second, we solve the equilibrium in Regime J with this linear information structure. Third, we derive tractable $q^\ddagger$ by comparing $\bar{a}_J$ with $a_C$, and $q^\dagger$ from comparing $U_{J}$ with $U_{C}$. Finally, we explore the effects of the verification and audit incentives on $q^\ddagger$ and $q^\dagger$ to complete the comparative statics.

\subsubsection{Step 1:}
Consider two linear information structures $\pi_{q_1}:\Theta \to \Delta(S_1)$ and $\pi_{q_2}:\Theta \to \Delta(S_2)$, where $\Theta = \{G, B\}$ and $S_1=S_2 = \{G, B, \emptyset\}$, and a conditional probability function $\gamma:S_1 \to \Delta(S_2)$. Specifically, $\gamma$ is defined as
\[
\gamma = 
\begin{pmatrix}
\frac{q_2}{q_1} & 0 & 1-\frac{q_2}{q_1} \\
0 & \frac{q_2}{q_1} & 1-\frac{q_2}{q_1} \\
0 & 0 & 1
\end{pmatrix}.
\]

For $j \in \{q_1, q_2\}$, the signal $\pi_j$ is defined by its conditional probability distributions
\begin{align*}
    \pi_j(G) &= (q_j,\ 0,\ 1-q_j), \\
    \pi_j(B) &= (0,\ q_j,\ 1-q_j),
\end{align*}
where $q_1, q_2 \in [0, 1]$ and $q_1> q_2$ without loss of generality. Rewriting these two information structures in matrix form leads to
\[
\pi_{q_1} = \begin{pmatrix}
q_1 & 0 & 1-q_1 \\
0 & q_1 & 1-q_1
\end{pmatrix},
\quad
\pi_{q_2} = \begin{pmatrix}
q_2 & 0 & 1-q_2 \\
0 & q_2 & 1-q_2
\end{pmatrix}.
\]
It is straightforward to verify $\pi_{q_2}=\pi_{q_1}\cdot\gamma$, which indicates that $\pi_{q_2}$ is a garbling of $\pi_{q_1}$ iff $q_1>q_2$. Therefore, by the Blackwell Theorem, $\pi_{q_1}$ Blackwell dominates $\pi_{q_2}$ iff $q_1>q_2$. 

\subsubsection{Step 2:}
Applying the linear information structure in Regime J by simplifying Equations \eqref{aw}, \eqref{bara}, and \eqref{betaJ} yields
\begin{align*}
a_J(\emptyset)&=i\Pr(\theta=B\mid r=g,s=\emptyset)=i\frac{1-\hat{\beta}_J}{2-\hat{\beta}_J},\\
a_J(G)&=0,a_J(B)=i,\bar{a}_J=qi+(1-q)i\frac{1-\hat{\beta}_J}{2-\hat{\beta}_J},\beta_J=\frac{l}{2}\bar{a}_J.
\end{align*}
Imposing the rational expectations requirement $\hat{\beta}_J=\beta_J$ gives
\begin{align}
\label{equilibriumJ_Linear}
2\beta_J^2 - (4 + i l)\beta_J + i l (q+1) = 0.
\end{align}
Solving Equation \eqref{equilibriumJ_Linear} leads to the following equilibrium $(\beta_J^q,\bar{a}_J^q)$:
\begin{align}
\beta_J^q&=\frac{4+il-\sqrt{16+i^2l^2-8ilq}}{4},\label{betaJinq}\\
\bar a_J^q&=\frac{4+il-\sqrt{16+i^2l^2-8ilq}}{2l}.\label{barainq}
\end{align}

\subsubsection{Step 3:}
Combining $a_C$ in \eqref{Exp: aC} and $\bar{a}_J^q$ in \eqref{barainq} yields a unique solution $q^\ddagger$, which is 
\begin{equation}
q^\ddagger=\frac{16-2il+6il^2-2i^2l^2-(4-il)m}{18il},\label{qddagger}
\end{equation}
where $m\equiv\sqrt{16 + 4il + 2il^2(2i-3)}$.

Use $U_J^q(\beta_J^q,\bar{a}_J^q)$ to denote the stakeholder's utility in Regime J with the linear information structure. Since $U_J^q(q=0)<U_{C}<U_{J}^q(q=1)$ and $\frac{dU_{J}^q(q)}{d q}>0$, the unique $q^\dagger$ can be directly derived by Step 1.

Simplifying $U_J^q=U_{C}$, we have
\begin{align}
108 i^2 l^2 q+ (4 - i l) (64 - i l (14 - (18 - 5 i) l))=8 (8 + i l (2 - (3 - 2 i) l)) m-27 i l (2 - (2 - i) l) n_q,\label{Exp: qdagger}
\end{align}
where $n_q\equiv\sqrt{16 + i l (i l - 8 q)}$. The solution of equation \eqref{Exp: qdagger} is exactly the $q^\dagger$, which is 
\begin{align}
q^\dagger=&\frac{1}{54 i^2 l^2} \Bigl(
-128 - 16 i^3 l^3 + 3 i^2 l^2 (-17 + 21 l) + 6 i l (1 + 12 l - 9 l^2) + 2 m^3 \notag\\
&- 3 \sqrt{3} \bigl(2 + (-2 + i) l\bigr) 
\sqrt{ i l \bigl( 128 + 12 i^2 (2 - 3 l) l^2 + 16 i^3 l^3 + 3 i l (25 - 6 l + 9 l^2) - 2 m^3 \bigr) } \Bigr).\label{qdaggerlinear}
\end{align}

\subsubsection{Step 4:}
\begin{align}
\frac{\partial q^\ddagger}{\partial i}
&=\frac{\left(64+(2-3l)(4il+i^2l^2)+4i^3l^3\right)-(16+2i^2l^2)m}{18i^2lm},\label{DqddaggerDi}\\
\frac{\partial q^\ddagger}{\partial l}
&=\frac{(4+2il-m)(32+4i^2l^2-12il^2)-12il(1-l)(4+il)}{36il^2m}.\label{DqddaggerDl}
\end{align}

For \eqref{DqddaggerDi}, the denominator $18i^2lm>0$. As both terms in the numerator are obviously positive, simplifying $\frac{\partial q^\ddagger}{\partial i}$ gives
\begin{align*}
\frac{\partial q^\ddagger}{\partial i}&\propto \left(64+(2-3l)(4il+i^2l^2)+4i^3l^3\right)^2-(16+2i^2l^2)^2m^2\\
&=i^2l^2(-1728-576l-252i^2l^2-108i^2l^3+288il+288il^2+144l^2+72il^3+9i^2l^4)\\
&<i^2l^2(-1728-576l-252i^2l^2-108i^2l^3+801)<0.
\end{align*}
The first two lines use basic algebraic rules. The first inequality in the last line holds since $0<i,l<1$, and the second inequality holds obviously. This proves $\frac{\partial q^\ddagger}{\partial i}<0$.

As for \eqref{DqddaggerDl}, the denominator $36il^2m>0$. Some algebraic manipulations yield
\begin{align*}
\frac{\partial q^\ddagger}{\partial l}&\propto (4+2il-m)(32+4i^2l^2-12il^2)-12il(1-l)(4+il)\\
&=3la_C(32+4i^2l^2-12il^2)-12il(1-l)(4+il)\\
&>\frac{6il(l+2)(32+4i^2l^2-12il^2)}{8+4il}-12il(1-l)(4+il)\\
&=\frac{18 i l^2 \bigl(8 - 4i - i l^2 + 2i l + i^2 l^2\bigr)}{2+il}>0.
\end{align*}
The second line uses \eqref{Exp: aC} that $4+2il-m=3la_C$. The inequality in the third line follows immediately from \eqref{equilibriumC} that $a_C=2i(l+2)/\bigl((8+4il)-3la_C\bigr)>2i(l+2)/(8+4il)$. The inequality in the last line holds since $8 - 4i - i l^2 + 2i l + i^2 l^2>3$. Thus, this proves $\frac{\partial q^\ddagger}{\partial l}>0$.

Similarly, the effect of verification and audit incentives on $q^\dagger$ can be calculated directly, as $q^\dagger$ is analytically tractable. Proofs of $\frac{\partial q^\dagger}{\partial i}<0$ and $\frac{\partial q^\dagger}{\partial l}>0$ are lengthy and follow procedures analogous to the above comparative statics for $q^\ddagger$, thus omitted.

\bibliographystyle{apalike}
\bibliography{Ref}

@article{lennox2010auditing,
  title={Auditing the auditors: Evidence on the recent reforms to the external monitoring of audit firms},
  author={Lennox, Clive and Pittman, Jeffrey},
  journal={Journal of Accounting and Economics},
  volume={49},
  number={1-2},
  pages={84--103},
  year={2010},
  publisher={Elsevier}
}

@article{lohlein2016peer,
  title={From peer review to {PCAOB} inspections: Regulating for audit quality in the US},
  author={L{\"o}hlein, Lukas},
  journal={Journal of Accounting Literature},
  volume={36},
  pages={28--47},
  year={2016},
  publisher={Elsevier}
}

@article{ehlen1996procedural,
  title={Procedural Fairness in the Peer and Quality Review Programs},
  author={Ehlen, Craig R and Welker, Robert B},
  journal={Auditing: A Journal of Practice \& Theory},
  volume={15},
  number={1},
  year={1996}
}

@article{wallace1994exploratory,
  title={An exploratory content analysis of terminology in public accounting firms’ responses to AICPA peer reviews},
  author={Wallace, WA and Cravens, KS},
  journal={Research in Accounting Regulation},
  volume={8},
  pages={3--32},
  year={1994}
}

@article{houston2013audit,
  title={Audit partner perceptions of post-audit review mechanisms: An examination of internal quality reviews and {PCAOB} inspections},
  author={Houston, Richard W and Stefaniak, Chad M},
  journal={Accounting Horizons},
  volume={27},
  number={1},
  pages={23--49},
  year={2013},
  publisher={American Accounting Association}
}

@book{von1934marktform,
  title={Marktform und Gleichgewicht},
  author={Stackelberg, Heinrich von},
  lccn={36009325},
  url={https://books.google.com.hk/books?id=wihBAAAAIAAJ},
  year={1934},
  publisher={J. Springer},
  address= {Wien und Berlin}
}

@article{ye2023theory,
  title={The Theory of Auditing Economics: Evidence and Suggestions for Future Research},
  author={Ye, Minlei},
  journal={Foundations and Trends (R) in Accounting},
  volume={18},
  number={3},
  pages={138--267},
  year={2023},
  publisher={now publishers}
}

@article{gao2019auditing,
  title={Auditing standards, professional judgment, and audit quality},
  author={Gao, Pingyang and Zhang, Gaoqing},
  journal={The Accounting Review},
  volume={94},
  number={6},
  pages={201--225},
  year={2019},
  publisher={American Accounting Association}
}

@article{deng2012auditors,
  title={Auditors’ liability, investments, and capital markets: A potential unintended consequence of the Sarbanes-Oxley Act},
  author={Deng, Mingcherng and Melumad, Nahum and Shibano, Toshi},
  journal={Journal of Accounting Research},
  volume={50},
  number={5},
  pages={1179--1215},
  year={2012},
  publisher={Wiley Online Library}
}

@article{patterson2007effects,
  title={The effects of Sarbanes-Oxley on auditing and internal control strength},
  author={Patterson, Evelyn R and Smith, J Reed},
  journal={The Accounting Review},
  volume={82},
  number={2},
  pages={427--455},
  year={2007}
}

@article{dye1993auditing,
  title={Auditing standards, legal liability, and auditor wealth},
  author={Dye, Ronald A},
  journal={Journal of Political Economy},
  volume={101},
  number={5},
  pages={887--914},
  year={1993},
  publisher={The University of Chicago Press}
}

@article{ye2013economics,
  title={The economics of setting auditing standards},
  author={Ye, Minlei and Simunic, Dan A},
  journal={Contemporary Accounting Research},
  volume={30},
  number={3},
  pages={1191--1215},
  year={2013},
  publisher={Wiley Online Library}
}

@article{defond2010should,
  title={How should the auditors be audited? Comparing the {PCAOB} inspections with the AICPA peer reviews},
  author={DeFond, Mark L},
  journal={Journal of Accounting and Economics},
  volume={49},
  number={1-2},
  pages={104--108},
  year={2010},
  publisher={Elsevier}
}

@article{defond2011effect,
  title={The effect of SOX on small auditor exits and audit quality},
  author={DeFond, Mark L and Lennox, Clive S},
  journal={Journal of Accounting and Economics},
  volume={52},
  number={1},
  pages={21--40},
  year={2011},
  publisher={Elsevier}
}

@article{chen2019effects,
  title={The effects of audit quality disclosure on audit effort and investment efficiency},
  author={Chen, Qi and Jiang, Xu and Zhang, Yun},
  journal={The Accounting Review},
  volume={94},
  number={4},
  pages={189--214},
  year={2019},
  publisher={American Accounting Association}
}

@article{gipper2020public,
  title={Public oversight and reporting credibility: Evidence from the {PCAOB} audit inspection regime},
  author={Gipper, Brandon and Leuz, Christian and Maffett, Mark},
  journal={The Review of Financial Studies},
  volume={33},
  number={10},
  pages={4532--4579},
  year={2020},
  publisher={Oxford University Press}
}

@article{lamoreaux2016does,
  title={Does {PCAOB} inspection access improve audit quality? An examination of foreign firms listed in the United States},
  author={Lamoreaux, Phillip T},
  journal={Journal of Accounting and Economics},
  volume={61},
  number={2-3},
  pages={313--337},
  year={2016},
  publisher={Elsevier}
}

@article{fung2017does,
  title={Does the {PCAOB} international inspection program improve audit quality for non-US-listed foreign clients?},
  author={Fung, Simon Yu Kit and Raman, KK and Zhu, Xindong Kevin},
  journal={Journal of Accounting and Economics},
  volume={64},
  number={1},
  pages={15--36},
  year={2017},
  publisher={Elsevier}
}

@article{gunny2013pcaob,
  title={{PCAOB} inspection reports and audit quality},
  author={Gunny, Katherine A and Zhang, Tracey Chunqi},
  journal={Journal of Accounting and Public Policy},
  volume={32},
  number={2},
  pages={136--160},
  year={2013},
  publisher={Elsevier}
}

@article{khurana2021pcaob,
  title={{PCAOB} inspections and the differential audit quality effect for Big 4 and non-Big 4 US auditors},
  author={Khurana, Inder K and Lundstrom, Nathan G and Raman, KK},
  journal={Contemporary Accounting Research},
  volume={38},
  number={1},
  pages={376--411},
  year={2021},
  publisher={Wiley Online Library}
}

@article{aobdia2017regulatory,
  title={Regulatory oversight and auditor market share},
  author={Aobdia, Daniel and Shroff, Nemit},
  journal={Journal of Accounting and Economics},
  volume={63},
  number={2-3},
  pages={262--287},
  year={2017},
  publisher={Elsevier}
}

@article{aobdia2018impact,
  title={The impact of the {PCAOB} individual engagement inspection process—Preliminary evidence},
  author={Aobdia, Daniel},
  journal={The Accounting Review},
  volume={93},
  number={4},
  pages={53--80},
  year={2018},
  publisher={American Accounting Association}
}

@article{daugherty2010pcaob,
  title={{PCAOB} inspections of smaller {CPA} firms: The perspective of inspected firms},
  author={Daugherty, Brian and Tervo, Wayne},
  journal={Accounting Horizons},
  volume={24},
  number={2},
  pages={189--219},
  year={2010}
}

@article{cunningham2019s,
  title={What's in a name? Initial evidence of US audit partner identification using difference-in-differences analyses},
  author={Cunningham, Lauren M and Li, Chan and Stein, Sarah E and Wright, Nicole S},
  journal={The Accounting Review},
  volume={94},
  number={5},
  pages={139--163},
  year={2019},
  publisher={American Accounting Association}
}

@article{burke2019audit,
  title={Audit partner identification and characteristics: Evidence from US Form AP filings},
  author={Burke, Jenna J and Hoitash, Rani and Hoitash, Udi},
  journal={Auditing: A Journal of Practice \& Theory},
  volume={38},
  number={3},
  pages={71--94},
  year={2019},
  publisher={American Accounting Association}
}

@article{christensen2024costs,
  author  = {Christensen, Brant E. and Newton, Nathan J. and Wilkins, Michael S.},
  title   = {Costs and Benefits of a Risk-Based {PCAOB} Inspection Regime},
  journal = {Accounting, Organizations and Society},
  year    = {2024},
  volume  = {112},
  pages   = {101552},
  doi     = {10.1016/j.aos.2024.101552}
}

@article{johnson2019auditors,
  author  = {Johnson, Lindsay M. and Keune, Marsha B. and Winchel, Jennifer},
  title   = {{U.S.} Auditors' Perceptions of the {PCAOB} Inspection Process: A Behavioral Examination},
  journal = {Contemporary Accounting Research},
  year    = {2019},
  volume  = {36},
  number  = {3},
  pages   = {1540--1574},
  doi     = {10.1111/1911-3846.12467}
}

@article{westermann2019pcaob,
  author  = {Westermann, Kimberly D. and Cohen, Jeffrey and Trompeter, Greg},
  title   = {{PCAOB} Inspections: Public Accounting Firms on ``Trial''},
  journal = {Contemporary Accounting Research},
  year    = {2019},
  volume  = {36},
  number  = {2},
  pages   = {694--731},
  doi     = {10.1111/1911-3846.12454}
}

@article{chen2025pcaob,
  author  = {Chen, Po-Chang and Moul, Charles and Reffett, Andrew},
  title   = {Do {PCAOB} Inspections Change the Effect of Litigation Risk on Audit Quality?},
  journal = {The International Journal of Accounting},
  year    = {2025},
  volume  = {60},
  number  = {1},
  pages   = {1--40},
  articleno = {2450019},
  doi     = {10.1142/S1094406024500197}
}
\end{document}